\def\year{2021}\relax
\documentclass[letterpaper]{article} 
\usepackage{aaai21}  
\usepackage{times}  
\usepackage{helvet} 
\usepackage{courier}  
\usepackage[hyphens]{url}  
\usepackage{graphicx} 
\usepackage{natbib}  
\usepackage{caption} 
\usepackage{amssymb,amsfonts,amsthm}
\newtheorem{myexample}{Example}
\newtheorem{mylemma}{Lemma}

\newtheorem{claim}{Claim} 
\newtheorem{theorem}{Theorem} 

\usepackage{stmaryrd}
\usepackage{mathrsfs}
\usepackage{paralist}
\usepackage{tabularx}
\usepackage{diagbox}
\usepackage{nicefrac}
\usepackage{mathtools}
\usepackage[switch]{lineno}  %

\usepackage{enumitem}
\setitemize{noitemsep,topsep=0pt,parsep=0pt,partopsep=0pt}

\usepackage{algorithm}
\usepackage[noend]{algorithmic}
\usepackage{tikz}
\usetikzlibrary{automata, positioning, calc, shapes, arrows, fit, decorations.pathmorphing, decorations.text}

\usepackage{boxedminipage}
\usepackage{xspace}
\usepackage{pgfplots}
\pgfplotsset{compat=1.10}
\newcommand{\pbDef}[3]{%
\noindent
\begin{center}
\begin{boxedminipage}{0.98 \columnwidth}
#1\\[5pt]
\begin{tabular}{l p{0.75 \columnwidth}}
Input: & #2\\
Question: & #3
\end{tabular}
\end{boxedminipage}
\end{center}
}

\usepackage{booktabs}
\usepackage{algorithm}
\usepackage{algorithmic}

 	\usepackage{varioref}
	\usepackage{graphicx,subfigure}
 \newcounter{rownumber}[figure] 
\renewcommand{\citet}[1]{\citeauthor{#1}~\shortcite{#1}}

\title{Optimal Kidney Exchange with Immunosuppressants}

\author { 
        Haris Aziz,\textsuperscript{\rm 1}
        \'{A}gnes Cseh,\textsuperscript{\rm 2,3}
        John P.\ Dickerson,\textsuperscript{\rm 4}
				Duncan C.\ McElfresh\textsuperscript{\rm 4}\\
}
\affiliations {
    \textsuperscript{\rm 1} UNSW Sydney and Data61 CSIRO\\
    \textsuperscript{\rm 2} Hasso-Plattner-Institute, University of Potsdam, Prof.-Dr.-Helmert-Str. 2-3, 14482 Potsdam, Germany\\
		\textsuperscript{\rm 3} Institute of Economics, Centre for Economic and Regional Studies, T\'{o}th K. u. 4, 1097 Budapest, Hungary\\
		\textsuperscript{\rm 4} University of Maryland\\ 
    haris.aziz@unsw.edu.au, agnes.cseh@hpi.de, john@cs.umd.edu, dmcelfre@umd.edu
}

\begin{document}

\maketitle

\begin{abstract}
	Algorithms for exchange of kidneys is one of the key successful applications in market design, artificial intelligence, and operations research. 
Potent immunosuppressant drugs suppress the body's ability to reject a transplanted organ up to the point that a transplant across blood- or tissue-type incompatibility becomes possible. In contrast to the standard  kidney exchange problem, we consider a setting that also involves the decision about which recipients receive from the limited supply of immunosuppressants that make them compatible with originally incompatible kidneys. We firstly present a general computational framework to model this problem. Our main contribution is a range of efficient algorithms that provide flexibility in terms of meeting meaningful objectives. Motivated by the current reality of kidney exchanges using sophisticated mathematical-programming-based clearing algorithms, we then present a general but scalable approach to optimal clearing with immunosuppression; we validate our approach on realistic data from a large fielded exchange.  
\end{abstract}

\section{Introduction}

The deployment of centralized matching algorithms for efficient exchange of donated kidneys is a major success story of market design~\citep{BHA19,BKM+19}. 
The theory and practice of kidney exchange have benefited from active research within artificial intelligence \citep[e.g.][]{ABS07a,MO15,MBD19a,MD18a,FDS17}.
The standard model for kidney exchange involves information about recipients' compatibility with kidneys in the market. A recipient can only be given a kidney that is compatible with the recipient. The goal is to enable exchanges of kidneys via a centralized algorithm to satisfy the maximum number of recipients. 

We consider a new kidney exchange model which has an interesting feature that is informed by significant technological advances in organ transplant. The technology concerns immunosuppressants which if given to a recipient can make her receptive to kidneys which she is not receptive to by default~\cite{MLK11}. We will refer to the model as \emph{Kidney Exchange with Immunosuppressants (KEI)}. 
Immunosuppressants (abbreviated as suppressants from here onwards) have been successfully used in Japan and Korea for several years, and increasingly being considered and utilized in other countries~\cite{HHC20a}, even though they are costly and may have side effects. Due to these costs or side effects, it is desirable to match as many recipients to kidneys while minimizing the number of recipients who are given immunosuppressants.

In this paper, the fundamental research problem that we explore is that of designing mechanisms for kidney exchange with suppressants that satisfy desirable computational, incentive and monotonicity properties. A naive way of using suppressants is to clear the classic kidney exchange market without using them and then give suppressant to the recipients who are left. However, there can be more efficient ways of giving suppressant to particular recipients and then implementing exchanges of kidneys to facilitate as many transplants as possible, {as we will demonstrate in Figure~\ref{fig:firstexample}}. At first sight, the two-stage and connected process of using foresight to first giving suppressants to suitable recipients and then finding a matching that satisfies suitable social objectives appears to be a complex problem. We design a flexible algorithmic approach for the problem. 

\paragraph{Contributions}
We formalize a \emph{general model} of KEI that features compatible, half-compatible, and incompatible kidneys, and which allows for allocations as a result of multi-way exchanges. We then initiate a computational study of kidney exchange with suppressants. Prior mechanism design work on the subject either only allows pairwise exchanges or focuses on a restricted model.

One of our central contributions is modeling 
important KEI problems in terms of an underlying graph with different classes of edges. One of the edge classes represents organ compatibility that is dependent on administering immunosuppressants. Depending on how we set the edge weights in the graph problem, we can find in polynomial time, allocations corresponding to several important objectives. The objectives include maximizing the total number of transplants and given that, maximizing the number of compatible transplants. Among the list of objectives captured by our algorithmic framework, we defer the choice of the exact objective to the policy-makers. 
	
Then, we focus on the problem where there is an \emph{upper bound on the number of suppressants} that can be used. We present a polynomial-time algorithm for maximizing the number of transplants for a restricted model that we refer to as the Silver Bullet model. In the model, once a recipient has been given a suppressant, then the recipient can take any kidney. For our general model in which certain kidneys are inherently incompatible, we show that the problem of maximizing the number of transplants reduces in polynomial-time to an interesting generalized matching problem whose complexity has been open for years. 

Finally, we present a flexible integer linear program (ILP) formulation that allows us to optimize objectives subject to bounds on the length of exchange cycles.  We validate that model on realistic data from a large, fielded kidney exchange in the United States, and show significant gains in the number of matches made even when the central clearinghouse is only able to use a small number of suppressants.

Some of the techniques that we use such as to capture strong individual rationality or handle pairwise exchanges etc. are of independent interest and can be applied to a host of other problems in matching markets. Although we present our model and result in the language of kidney exchange and suppressants, our model and algorithms also apply to any exchange model in which agents have trichotomous preferences~\cite{MW18} and for any half-compatible match to materialize, the social designer needs to use some resource such as money to facilitate such a match. The goal is to implement desirable exchanges subject to minimum use of additional resources.


\section{Related Work}

Kidney exchange is one of the major research topics in matching market design~\cite{ABS07a,AsRo20a,Hatf05a,BMR09a,DPS14,LLM+19,RSU05a,SoUn10a}. In many of the papers, the algorithms only allow exchange cycles of limited size due to logistical and other constraints. In this paper, we first allow exchange cycles of any size, and then discuss the bounded case. Note that for any exchange cycle bounds that are three or more, even the kidney exchange problem in the traditional model without suppressants is NP-hard~\cite{ABS07a}.

The use of suppressants to facilitate more efficient kidney exchange has been discussed in medical circles~(see, e.g. \citet{AOH+18a}). The two market design papers directly relevant to our work are the ones where kidney exchange with suppressants has been mathematically modeled~\cite{HHC20a,AnKr20}. \citet{HHC20a} prove a couple of impossibility results as well as an exchange mechanism with some desirable monotonicity properties. They assume that once suppressants are administered to a recipient, she can take a kidney from any donor. We consider a more general model in which half-compatibility is specific to particular recipient-donor pairs. 

\citet{AnKr20} consider a model more general than that of \citet{HHC20a} in which only certain donor-recipient pairs can be made compatible after giving suppressants to the recipient. They focus on \emph{pairwise} kidney exchange~\cite{RSU05} and demonstrate through experiments that adding suppressant treatment to the pairwise exchange model results in a larger increase in transplant numbers than allowing short cycles.
Considering both cycles and suppressants was discussed as important future work by \citet{AnKr20}. Our paper presents experimental results in this setting. 
		
\section{Model and Concepts}
\label{se:model}

A kidney exchange market is a tuple $(R,D,C,H,I)$ where $R=\{r_1, r_2, \ldots, r_n\}$ is a set of $n$ recipients (agents) and $D$ is the set of donors. Some recipients and donors come in pairs; others come single. 
Generous donors who offer their kidney to the pool instead of a specific recipient in it are called \emph{altruistic donors}. 
 
Each recipient $r_i$ partitions the donors $D$ into sets $C_i$, $H_i$, and~$I_i$. The set $C_i$ is the set of donors whom recipient $r_i$ is compatible with. The set $H_i$ is the set of donors $r_i$ is half-compatible with. Half-compatibility means that if a suppressant is given to $r_i$, then $r_i$ can accept a kidney from any donor in~$H_i$. Donors in the set $I_i$ are incompatible with recipient $r_i$ even if $r_i$ is given a suppressant. These partitions at each recipient form the collections of sets $C=(C_1,\ldots, C_n)$, $H=(H_1,\ldots, H_n)$, and $I=(I_1,\ldots, I_n)$ in the input. 

An \emph{allocation} assigns each recipient $r_i$ at most one donor who is either in $C_i$ or in~$H_i$. Recipients who are assigned a half-compatible donor receive suppressants.

We consider three models.
\begin{enumerate}
	\item \textbf{BM (Baseline model)}: for each $r_i\in R$, $H_i=\emptyset$.
	\item \textbf{SBM (Silver Bullet model)}: for each $r_i\in R$, $I_i=\emptyset$.
	\item \textbf{GM (General model)}.
\end{enumerate}

The baseline model coincides with the traditional kidney exchange model in which suppressants are not considered. SBM is the model in which we assume that if a recipient is given a suppressant then she will be able to receive any kidney in the market~\cite{HHC20a}. 
GM is the general model that also allows for some kidneys being inherently incompatible for a recipient even if she has been given suppressants. Unless specified, we will focus on GM. In some cases, we will present some results that hold for the Silver Bullet model (SBM). The SBM assumption was made by \citet{HHC20a} so we keep it as an important intermediate model between the baseline model and general model.  Except for Theorems~\ref{prop:maxkei} and \ref{prop:$h$-MaxKEI}, all of our results and discussions hold for the general model. 

As far as a recipient or the social designer is concerned, there are two types of preferences. We will treat matching with incompatible donors to be infeasible.

\begin{enumerate}
	\item \textbf{Coarse preferences}: a recipient is indifferent between a compatible donor and a half-compatible donor with a suppressant, and prefers both options over no transplant at all.
	\item \textbf{Refined preferences}: a recipient prefers compatible donors over half-compatible donors along with a suppressant, which are preferred over no transplant at all.
	\end{enumerate}

Coarse preferences have the underlying assumption that a recipient has no significant cost (in terms of money or side-effects) when receiving a half-compatible kidney.
Based on the preference relation one can define concepts such as Pareto optimality. \citet{HHC20a} considered SBM and coarse preferences. They consider refined preferences when defining a monotonicity property. \citet{AnKr20} considered GM and refined preferences. 

A recipient who is assigned a  compatible donor or a half-compatible donor along with a suppressant is referred to as \emph{satisfied}. Our general goal is to satisfy the maximum number of recipients while minimizing the need of suppressants. We will consider the following feasibility condition for all allocations, which captures a natural individual rationality requirement: either a recipient donates her donor's kidney to the market and gets a strict improvement or she and her donated kidney are not part of any allocation. We will refer to this condition as \emph{strong individual rationality (strong-IR).} 
A recipient who enters the market with a half-compatible donor improves her situation if she is assigned to her own or another half-compatible donor along with a suppressant. Also, strong-IR implies that a donor arriving in a pair with a recipient will only donate a kidney if her recipient also receives one. A weaker requirement is \emph{individual rationality (IR)} whereby no recipient whose own donor is compatible ends up with no transplant {or a half compatible kidney}.

 \begin{myexample}
 Consider a kidney exchange problem in which there are three recipients $r_1, r_2, r_3$ with corresponding donors $d_1, d_2, d_3$. No recipient's donor has a kidney compatible with the recipient. Recipient $r_1$ finds the kidney of $d_2$ compatible{, while $d_3 \in H_2$ and $d_1 \in H_3$}. The problem is captured in Figure~\ref{fig:firstexample}.

 If suppressants are not allowed, then no recipient will be able to get a kidney without violating strong-IR. This remains the case if only one suppressant is allowed. Suppose now that the system has 2 suppressants available. In that case, one suppressant can be given $r_2$ and another to~$r_3$. Then $r_1$ can take a compatible kidney of $d_2$, $r_2$ takes a half-compatible kidney of $d_1$ and $r_3$ takes a half-compatible kidney of~$d_3$.
 \end{myexample}
   	    	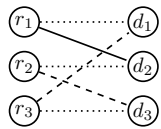
\begin{figure}[h!]
					\begin{center}
					\resizebox{.12\textwidth}{!}{
   	    	\begin{tikzpicture}[-,>=stealth',shorten >=1pt,auto,node distance=3cm,
   	    	        thick,main node/.style={circle,fill=white!20,draw,minimum size=0.5cm,inner sep=0pt]}, scale=0.5]

   	    	    \node[main node] at (0, 0) (1)  {$r_1$};
   	    	    \node[main node] at (0, -1.5) (2)  {$r_2$};
   	    	    \node[main node] at (0, -3) (3)  {$r_3$};

   	    	    \node[main node] at (4, 0) (d1)  {$d_1$};
   	    	    \node[main node] at (4, -1.5) (d2)  {$d_2$};
   	    	    \node[main node] at (4, -3) (d3)  {$d_3$};

   	       \draw[dotted] (1) -- (d1);
   	        \draw[dotted] (2) -- (d2);
   	        \draw[dotted] (3) -- (d3);
	        
   	           \draw[thick] (1) -- (d2);
   		   \draw[dashed] (2) -- (d3);
		   
   		     \draw[dashed] (3) -- (d1);

   	    	\end{tikzpicture}
					}
   	    	\end{center}
   		\caption{A bipartite matching view of KEI. Dashed lines indicate half-compatible edges. Solid edges indicate compatibility edges. Dotted edges indicate a recipient-donor pair.}
		\label{fig:firstexample}
   	\end{figure}

\section{A General Graph Theoretic Approach}
\label{sec:model}
We construct a general bipartite matching based model capturing the most basic features of kidney exchange markets. It guarantees that each donor gives at most one kidney, each recipient receives at most one kidney, and the donor in a pair is only part of the exchange if her recipient receives a kidney. This framework gives us a set of feasible solutions for the problem. Then, by adding edge weights to the graph and finding a maximum weight matching, an optimal solution can be calculated. We specify a set of possible edge weights that can serve a large variety of goals of the decision maker, such as cost-efficiency or saving as many lives as possible.

\subsection{Matching Model}
\label{se:mm}

We build a bipartite graph to the instance $(R,D,C,H,I)$, see Figure~\ref{fi:network}. For convenience, we distinguish between recipients with and without a related donor, who will form the sets $R_2$ and $R_1$, respectively. Analogously, $D_1$ is the set of altruistic donors, while donors in $D_2$ enter the market along with their related recipient in~$R_2$. The vertices of our graph are of three types.
\begin{itemize}
	\item To each recipient $r_i \in R$, we construct a recipient vertex~$r_i$.
	\item  To each donor $d_i \in D$, we construct a donor vertex~$d_i$.
	\item We construct a dummy donor vertex $d_j$ for each recipient $r_j \in R_1$, and a dummy  recipient vertex $r_j$ for each donor $d_j \in D_1$.
\end{itemize}

If a donor-recipient pair who applies for the scheme together, then the recipient and donor are given the same index: we refer to them as $r_i$ and $d_i$ for some fixed~$i$. Recipients without a donor share the same index with their dummy donor vertex, and an analogous notation is applied for altruistic donors and their dummy counterparts. Dummy donors form the set $D_0$, while dummy recipients form the set~$R_0$.

The edges of the graph are as follows.

\begin{itemize}
	\item Each recipient $r_i \in R$ is connected to the donor with the same index $d_i \in D$ via a \emph{private edge}.
	\item Each dummy recipient $r_j \in R_0$ is connected to all donors via \emph{dummy edges}.
	\item A donor $d_i$ has a \emph{compatible} edge to a recipient $r_j$, where $i$ might be equal to $j$, if $d_i \in C_i$.
	\item A donor $d_i$ has a \emph{half-compatible} edge to a recipient $r_j$, where $i$ might be equal to $j$, if $d_i \in H_i$.
\end{itemize}

The four kinds of edges will play distinct roles when assigning weights to them. Private and dummy edges represent no transplant, while compatible and half-compatible edges stand for compatible and half-compatible transplants. 
Notice that a recipient and her donor forming a half-compatible (or compatible) pair are connected by two parallel edges, one private and one half-compatible (or compatible).

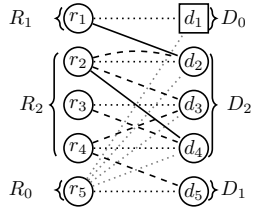
\begin{figure}[h!]
\begin{center}
\resizebox{.2\textwidth}{!}{
\begin{tikzpicture}[-,>=stealth',shorten >=1pt,auto,node distance=3cm, 	    	        thick,main node/.style={circle,fill=white!20,draw,minimum size=0.5cm,inner sep=0pt},
dummy node/.style={rectangle,fill=white!20,draw,minimum size=0.5cm,inner sep=0pt}, scale=0.5]

   	    	    \node[main node] at (0, 0) (1)  {$r_1$};
   	    	    \node[main node] at (0, -1.5) (2)  {$r_2$};
   	    	    \node[main node] at (0, -3) (3)  {$r_3$};
							\node[main node] at (0, -4.5) (4)  {$r_4$};
							\node[main node] at (0, -6) (5)  {$r_5$};

   	    	    \node[dummy node] at (4, 0) (d1)  {$d_1$};
   	    	    \node[main node] at (4, -1.5) (d2)  {$d_2$};
   	    	    \node[main node] at (4, -3) (d3)  {$d_3$};
							\node[main node] at (4, -4.5) (d4)  {$d_4$};
							\node[main node] at (4, -6) (d5)  {$d_5$};
	     
						\draw[dotted, gray] (5) -- (d1);
						\draw[dotted, gray] (5) -- (d2);
						\draw[dotted, gray] (5) -- (d3);
						\draw[dotted, gray] (5) -- (d4);
						\draw[dotted] (5) -- (d5); 
						\draw[dotted] (1) -- (d1);
   	        \draw[dotted] (2) -- (d2);
   	        \draw[dotted] (3) -- (d3);
						\draw[dotted] (4) -- (d4);
						\draw[thick] (2) -- (d4);
						\draw[dashed] (3) -- (d4);
						\draw[dashed] (4) -- (d5);
						\draw[dashed] (4) -- (d3);
						\draw[thick] (1) -- (d2);
						\draw[dashed] (2) -- (d3);
						\draw [dashed] (2) to[out=15,in=165] (d2);
				
				\draw [decorate,decoration={brace,amplitude=4pt,mirror},yshift=0pt]
(4.5,-0.4) -- (4.5,0.4) node [black,midway,xshift=0.8cm] {$D_0$};
					\draw [decorate,decoration={brace,amplitude=4pt,mirror},yshift=0pt]
(4.5,-4.8) -- (4.5,-1.2) node [black,midway,xshift=0.9cm] {$D_2$};
					\draw [decorate,decoration={brace,amplitude=4pt,mirror},yshift=0pt]
(4.5,-6.3) -- (4.5,-5.6) node [black,midway,xshift=0.8cm] {$D_1$};
\draw [decorate,decoration={brace,amplitude=4pt},yshift=0pt]
(-0.5,-0.4) -- (-0.5,0.4) node [black,midway,xshift=-0.4cm] {$R_1$};
\draw [decorate,decoration={brace,amplitude=4pt,mirror},xshift=-0.2cm]
(-0.5,-1.2) -- (-0.5,-4.8) node [black,midway,xshift=-0.8cm] {$R_2$};
\draw [decorate,decoration={brace,amplitude=4pt},yshift=0pt]
(-0.5,-6.3) -- (-0.5,-5.6) node [black,midway,xshift=-0.4cm] {$R_0$};
   	    	\end{tikzpicture}
					}
   	    	\end{center}
   		\caption{Example instance for our bipartite graph. Here, $R_1 = \left\{ r_1 \right\}$, and thus, $d_1$ is a dummy donor. The only altruistic donor is $d_5$, forming set $D_1$, and her dummy recipient is~$r_5$. Dotted black edges are private, dotted gray edges are dummy, dashed edges mark half-compatible donations, and finally, solid edges mark compatible donations.}
   		\label{fi:network}
   	    	\end{figure}

Our goal is to calculate a perfect matching in the constructed graph. A matching and the corresponding allocation are in trivial one-to-one correspondence with each other. A recipient matched along her private or dummy edge represents no transplant. The matching property ensures that each recipient in $R_1 \cup R_2$ receives one kidney at most, and each donor in $D_1 \cup D_2$ also donates one kidney at most. Since a recipient $r_i \in R_2$ is only connected to $d_i \cup C_i \cup H_i$, and we restrict our attention to perfect matchings only, $r_i$ is either satisfied or she participates in no transplant, keeping her donor~$d_i$. Perfectness thus ensures the following natural consequence of strong-IR: either a recipient uses her donor's kidney and gets a strict improvement, or she and her donor are not part of the allocation. 

To guarantee that a compatible pair only participates in a pairwise exchange or a cycle if and only if the recipient receives a compatible kidney, we only need to delete the edges running from the recipient to all half-compatible donors. This could be a natural requirement from a compatible recipient-donor pair who enter the market together---which actually often happens in practice, for example in the two largest exchange pools in Europe, in the Netherlands and in the UK~\cite{BHA19}.

\subsection{Objectives}

\begin{table}
	\centering
	\scalebox{0.64}{
		\begin{tabular}{clc c c}
		\toprule
			&\textbf{Objective function}& \textbf{compatible} & \textbf{half-compatible} & \textbf{private}\\ \midrule
      \setcounter{rownumber}{0}\refstepcounter{rownumber}\label{row:maxtr}1 & $TR$ (coarse preferences) & $1$ & $1$ & 0 \\ \midrule
			\refstepcounter{rownumber}\label{row:maxcomp}2 & $CO$ (BM, where $HC=0$) & $1$ & $-\infty$ & 0\\ \midrule
			\refstepcounter{rownumber}\label{row:maxmax}3 & $(CO,TR)$ & $N$ & $1$ & 0\\ \midrule
			\refstepcounter{rownumber}\label{row:minmax}4 & $(CO,-HC)$ & $N$ & $-1$ & 0\\ \midrule
			\refstepcounter{rownumber}\label{row:maxmax2}5 &  $(TR,-HC) \sim (TR, CO)$ & $N$ & $N-1$ & 0\\ \midrule
			\refstepcounter{rownumber}\label{row:mincost}6 & cost-optimal & compatible gain & half-compatible gain & waiting gain\\ \bottomrule
		\end{tabular}
		}
				\caption{The variety of possible weight functions serving different goals.}
				\label{ta:costs}
\end{table}

We offer a variety of different weight functions defined on the edges of our graph. Each weight function serves a justifiable goal, as we argue later.

Table~\ref{ta:costs} summarizes the options for defining the weight function on each edge $(r_i,d_j)$, depending on the type of the edge. Dummy edges always carry zero weight, therefore they are omitted from the table. We assume $N$ to be a sufficiently large integer, $n$ for example. In an allocation, we denote the number of recipients receiving a kidney from a compatible donor by $CO$, the number of recipients receiving a kidney from a half-compatible donor by $HC$, while the total number of recipients receiving a kidney by $TR=CO+HC$. Our objective functions are to be maximized in the lexicographic sense, e.g.\ $(TR,-HC)$ maximizes the number of transplants in total, and subject to this, it minimizes the number of half-compatible transplants.

Our objective functions can achieve the following.
\begin{enumerate}[noitemsep,topsep=0pt,parsep=0pt,partopsep=0pt]
	\item The number of transplants is maximized if each pair chosen for surgery contributes weight~$1$, while no transplant adds no weight to the matching.
	\item The number of transplants is maximized in the baseline model, if half-compatible donations are forbidden due to their infinitely large negative weight, and each compatible donation contributes weight~$1$.
	\item If a compatible transplant carries a  larger weight than the weight of all half-compatible transplants that can be carried out, then the main goal is to maximize the number of compatible donations. Since half-compatible donations do carry some small weight, their number will be maximized, but only subject to the first objective.
	\item Since half-compatible donations now carry a small negative weight, they are only to be planned if they enable extra compatible transplants. However, any number of half-compatible donations are welcome if they make only one more compatible donation happen, because we gain a lot in our objective function by adding $N$ just one more time to it.
	\item If compatible and half-compatible donations both carry a large weight, but the latter ones are somewhat less profitable, then the maximum number of donations will be calculated, and subject to this, as few half-compatible donations will be planned as possible.
	\item The most general version is when we set an arbitrary, possibly negative weight to each transplant. This weight can express the expected utility in terms of life expectancy, risks, healthcare savings, and it can differ for each pair. This objective is thus able to replace the trichotomous metric by a finely scaled one. Private edges represent withdrawal from donation, which can also be expressed in utilities, for example as loss due to health deterioration. On the positive side, sparing an exceptionally valuable donor in order to wait for a better match in the next round is also an entirely realistic scenario. A max weight solution corresponds to a maximum utility allocation.
\end{enumerate}
 
\begin{theorem}
\label{th:polytime_alg}
	For each of the objectives 1) to 6), there exists a strongly polynomial-time algorithm to find an allocation achieving those objectives. 
\end{theorem}
\begin{proof}
Our goal is to compute a maximum weight perfect matching in the graph. A maximum weight matching can be computed in strongly polynomial time~\cite{Mun57}. To take care of perfect\-ness, or equivalently, maximum size, one only needs to apply the standard weight modification~\cite{KV12}, in which each edge gets an additional uniform weight that is larger than the sum of the original weights in any matching. For this uniform addition, $n \cdot w_{\text{max}}(e) +1$ suffices.

$$w'(e) := w(e) + n \cdot w_{\text{max}}(e) +1$$

The weight of matching $M$ is thus $w'(M) := w(M) + (n \cdot w_{\text{max}}(e) +1) \cdot |M|$. Since $(n \cdot w_{\text{max}}(e) +1) \cdot |M| >  w(M)$ for any matching $M$, matchings of larger size are bound to have a larger weight as well, thus $w'(M)$ is maximized by a perfect matching. Within the set of perfect matchings, $(n \cdot w_{\text{max}}(e) +1) \cdot |M| = (n \cdot w_{\text{max}}(e) +1) \cdot n$ is identical, and thus $w'$ is maximized in the maximum weight matching according to~$w$.
\end{proof}

\section{Fixed upper bound on $HC$}
\label{se:bound}
Suppressants are highly useful to allow half-compatible kidneys to be allocated. However, they are not only extremely expensive but they also have undesirable side-effects. Given these issues, the market designer may wish to specify a fixed upper quota on $HC=h$, and wishes to maximize the number of transplants subject to~$h$. We show that even with an upper bound, we can solve the following central problem.

\pbDef{\textsc{$h$-AllKEI}}{\textsc{KEI} instance $G=(R,D,C,H,I)$ and integer $h$.}{ Is there an allocation satisfying all the recipients with at most $h$ recipients using suppressants?}

\begin{theorem}
	\textsc{$h$-AllKEI} can be solved in polynomial time even for the general model. 
	\end{theorem}
	\begin{proof}
		Construct the corresponding graph as defined in section on the matching model, with vertex sets $R$ and~$D$. 
		For the edges, here we only keep the edges of vertices in $R_0$, compatible, and half-compatible edges. In particular, we delete the private edges between same-index couples in $(R_1 \cup R_2) \times (D_0 \cup D_2)$.
		
		We want to check whether there exists an allocation such that every single recipient in $R_1 \cup R_2$ gets either a compatible or a half-compatible kidney, and at most $h$ of them receives a half-compatible kidney. In graph-theoretic terms, this question translates to deciding whether a perfect matching $M$ exists in the constructed graph, so that $M$ contains at most $h$ edges from the special edge set $E'$ of half-compatible edges. This question can be answered by solving a simple weighted perfect matching problem. In the reduced graph, we give each compatible edge weight 1, and to all other edges, weight~0. 
		The weight of any perfect matching $M$ is $n - |M\cap E'|$. The maximum weight matching in this instance has weight at least $n - h$ if and only if there is an allocation using at most $h$ suppressants.
		\end{proof}
		
		%
		%
		
Our next problem, \textsc{$h$-MaxKEI} is a more general version of \textsc{$h$-AllKEI}: 
	\pbDef{\textsc{$h$-MaxKEI}}{KEI instance $(R,D,C,H,I)$ and integers $t$ and $h$.}{Is there an allocation giving at least $t$ recipients a compatible donor with at most $h$ recipients using suppressants?}
	
	Our next result is that \textsc{$h$-MaxKEI} can be solved in polynomial time in the Silver Bullet model.
	
\begin{theorem}\label{prop:maxkei}
	\textsc{$h$-MaxKEI} can be solved in polynomial time in the Silver Bullet model.
	\end{theorem}
	\begin{proof}
If we have a model that excludes incompatibility, as the Silver Bullet model, then a  modification of the constructed graph solves this problem. We assume that each recipient in $R_1 \cup R_2$ is connected to each donor in $D_1 \cup D_2$ either via a half-compatible or via a compatible edge. Besides these edges, private and dummy edges are also present.

The modification of the graph is as follows. The goal is to decompose each half-compatible edge into a set of paths, and then lead these paths through a gadget that will regulate the maximum number of used half-compatible edges through its size. First we add this gadget, which consists of $2h$ new vertices in sets $A$ and $B$, and a set of $h$ disjoint edges of weight~0 between them: $\left\{(a_1,b_1), (a_2, b_2) \ldots, (a_h, b_h) \right\}$. Then we replace each edge $(r_i,d_j)$ in the half-compatible class by a set of edges connecting $r_i$ to each of $a_1, a_2, \ldots, a_h$, and $d_j$ to each of $b_1, b_2, \ldots, b_h$. The weight on these edges are set to be half of the original weight of the replaced half-compatible edges. The rest of the graph remains unchanged. Notice that vertex sets $R \cup B$ and $D \cup A$ build a bipartition of the new graph.

For an example, see Figure~\ref{fi:bound_on_HC}. The instance originates from our earlier example instance from Figure~\ref{fi:network}, with $h=2$. The difference from that instance is that while $(r_1,d_2)$ and $(r_2,d_4)$ are compatible edges as before, all other recipient-donor pairs are half-compatible unless the recipient or the donor is a dummy, so that the input suits the Silver Bullet model. Figure~\ref{fi:bound_on_HC} depicts the graph after vertex sets $A, B$, and the gadget on them are added to it.

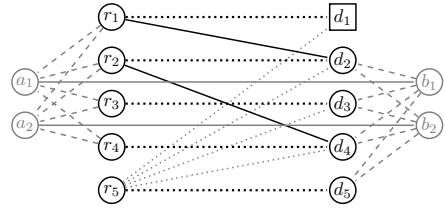
\begin{figure}[t]
\begin{center}
\resizebox{.33\textwidth}{!}{
\begin{tikzpicture}[-,>=stealth',shorten >=1pt,auto,node distance=3cm, 	    	        thick,
main node/.style={circle,fill=white!20,draw,minimum size=0.5cm,inner sep=0pt},
dummy node/.style={rectangle,fill=white!20,draw,minimum size=0.5cm,inner sep=0pt},
gadget node/.style={circle,gray,fill=white!20,draw,minimum size=0.5cm,inner sep=0pt},scale=0.55]

   	    	    \node[main node] at (0, 0) (1)  {$r_1$};
   	    	    \node[main node] at (0, -1.5) (2)  {$r_2$};
   	    	    \node[main node] at (0, -3) (3)  {$r_3$};
							\node[main node] at (0, -4.5) (4)  {$r_4$};
							\node[main node] at (0, -6) (5)  {$r_5$};

   	    	    \node[dummy node] at (8, 0) (d1)  {$d_1$};
   	    	    \node[main node] at (8, -1.5) (d2)  {$d_2$};
   	    	    \node[main node] at (8, -3) (d3)  {$d_3$};
							\node[main node] at (8, -4.5) (d4)  {$d_4$};
							\node[main node] at (8, -6) (d5)  {$d_5$};
	     
   	      \draw[dotted, gray] (5) -- (d1);
					\draw[dotted, gray] (5) -- (d2);
					\draw[dotted, gray] (5) -- (d3);
					\draw[dotted, gray] (5) -- (d4);
					\draw[dotted, very thick] (5) -- (d5);
					\draw[dotted, very thick] (1) -- (d1);
   	        \draw[dotted, very thick] (2) -- (d2);
   	        \draw[dotted, very thick] (3) -- (d3);
						\draw[dotted, very thick] (4) -- (d4);
						\draw[thick] (2) -- (d4);
   	        \draw[thick] (1) -- (d2);
   		    \node[gadget node] at (-3, -2.25) (a1)  {$a_1$};
					\node[gadget node] at (-3, -3.75) (a2)  {$a_2$};
	       \node[gadget node] at (11, -2.25) (b1)  {$b_1$};
					\node[gadget node] at (11, -3.75) (b2)  {$b_2$};
					\draw[thick, gray] (a1) -- (b1);
					\draw[thick, gray] (a2) -- (b2);
					\draw[dashed, gray] (1) -- (a1);
					\draw[dashed, gray] (2) -- (a1);
					\draw[dashed, gray] (3) -- (a1);
					\draw[dashed, gray] (4) -- (a1);
					\draw[dashed, gray] (1) -- (a2);
					\draw[dashed, gray] (2) -- (a2);
					\draw[dashed, gray] (3) -- (a2);
					\draw[dashed, gray] (4) -- (a2);
					\draw[dashed, gray] (b1) -- (d2);
					\draw[dashed, gray] (b1) -- (d3);
					\draw[dashed, gray] (b1) -- (d4);
					\draw[dashed, gray] (b1) -- (d5);
					\draw[dashed, gray] (b2) -- (d2);
					\draw[dashed, gray] (b2) -- (d3);
					\draw[dashed, gray] (b2) -- (d4);
					\draw[dashed, gray] (b2) -- (d5);
   	    	\end{tikzpicture}
					}
   	    	\end{center}
   		\caption{Substituting half-compatible edges by a gadget allowing at most 2 half-compatible donations. Compatible edges $(r_1,d_2)$, $(r_2,d_4)$, and the dotted private/dummy edges remain intact.}
   		\label{fi:bound_on_HC}
   	    	\end{figure}

\begin{claim}
A maximum weight perfect matching in the above-described graph corresponds to an allocation maximizing the weight subject to $HC \leq h$, if half-compatible donations are less desirable than compatible donations according to the weight function.
\end{claim}
\noindent\textit{\proofname.} Due to the size of the gadget, no perfect matching allows more than $h$ vertices in $R$ to be matched along their edges to the gadget. Moreover, the number of vertices in $R$ that are matched to a vertex in $A$ equals the number of vertices in $D$ that are matched to a vertex in $B$, because a perfect matching covers all vertices in the gadget. These vertices in $R$ and $D$ 
 will be the agents participating in half-compatible donations. Due to the assumptions of the Silver Bullet model, any perfect matching on them is a set of executable transplants. The rest of the transplants are chosen based on the maximum weight matching criterion.

Notice that it is possible that a donor in $D_1 \cup D_2$ and a recipient in $R_1 \cup R_2$ are connected via a 3-path through the gadget and via a direct compatible edge as well, but for all weight functions where half-compatible donations are less desirable than compatible donations (all our weight functions except for~\ref{row:maxtr} and possibly~\ref{row:mincost}), the path will carry the lower weight. Therefore, no perfect matching using such edges in the gadget can be of maximum weight.
\qedhere\qedsymbol
	\end{proof}
	
	This construction in the proof of Theorem~\ref{prop:maxkei} answers a question more general  than \textsc{$h$-MaxKEI}. It actually decides whether there is an allocation of weight at least $t$ while using at most $h$ suppressants.

	\begin{theorem}\label{prop:$h$-MaxKEI}
		In the Silver Bullet model, a maximum weight strong-IR allocation can be computed in polynomial-time even if there is an upper bound on the number of suppressants that can be used.
		\end{theorem}

Next, we identify connections of \textsc{$h$-MaxKEI} with a special case of budgeted matching, a well-studied graph problem of unknown complexity.

		\pbDef{{\textsc{Unit-Cost Budgeted Matching}}}{Bipartite graph $G=(A \cup B, E)$, $E' \subseteq E$, edge weights, and integers $h$ and~$t$.}{Is there a maximum weight matching $M$ of weight at least $t$ such that $|M\cap E'| \leq h$?}
		
\textsc{Unit-Cost Budgeted Matching} is a restricted variant of \textsc{Budgeted Matching}, where in addition to the edge weights, edge costs $c(e)$ are also present, and the budget $|M\cap E'| \leq h$ is substituted by $c(M) \leq h$. If costs are 0 or 1, then \textsc{Budgeted Matching} is identical to \textsc{Unit-Cost Budgeted Matching}, where $E'$ is the set of edges with cost~1.

\textsc{Unit-Cost Budgeted Matching} admits a PTAS~\cite{BBGS11,MS12}. \citet{BBGS11} observe that for polynomial weights and costs (here we set the costs to be 1), \textsc{Budgeted Matching} is very unlikely to be NP-hard, because it would imply RP=NP. However, after several decades, the problem of finding a deterministic algorithm to solve this problem is still open. We now show how \textsc{$h$-MaxKEI} reduces to \textsc{Unit-Cost Budgeted Matching}. 

\begin{mylemma}
	\textsc{$h$-MaxKEI} polynomial-time reduces to \textsc{Unit-Cost Budgeted Matching}. 
\end{mylemma}
	\begin{proof}
We set $E'$ to be the set of half-compatible edges, while $G$ and the upper bound $h$ are identical in the two problems. To make sure that the maximum weight matching in \textsc{Unit-Cost Budgeted Matching} is perfect, we modify the edge weights $w(e)$ from \textsc{$h$-MaxKEI} in an analogous manner to our method in the proof of Theorem~\ref{th:polytime_alg}:
$w'(e) := w(e) + n.$

The weight of matching $M$ is thus $w'(M) := w(M) + n |M|$. For $w(e) \leq 1$, matchings of larger size have a larger weight as well, thus $w'(M)$ is maximized by a perfect matching. Within the set of perfect matchings, $w'$ is maximized in the maximum weight matching according to~$w$.
\end{proof}

Regarding parametrized complexity, our trivial parameter is $h$, the number of suppressants available. If $h$ is small, then one can try which $h$ half-compatible edges are used, and then search for a maximum weight allocation in the rest of the instance built out of compatible and private edges only.

\section{Restrictions on the Exchange Cycle Length}
\label{se:bounded}

In kidney exchange, the length of the exchange cycles is typically required to be small for logistical reasons and to reduce the risk of a cycle being disrupted if someone backs out of the exchange. In this section, we focus our attention to short exchange cycles. 


Maximizing kidney exchange under the restriction on the size of the exchange cycles is NP-hard~\cite{ABS07a,BMR09a}. A practical approach to solving the problem involves formulating it as an ILP (Integer Linear Program).\footnote{Additionally, we address the case of \emph{pairwise exchange} in the supplemental material, and give a polynomial-time clearing algorithm for just that special case.}  We present an ILP based on PICEF~\citep{Dickerson16:Position}, given in (\ref{eq:ilp}) below.

First we construct the graph to the instance as described in the section on the matching model. Edges are equipped with the edge weight $w(e)$ serving any chosen objective in Table~\ref{ta:costs}.  To construct the corresponding ILP, we create the following binary variables:
\begin{itemize}
    \item $y_{ek}$: $1$ if edge $e$ is matched at position $k$ in a chain, and $0$ otherwise
    \item $z_c$: $1$ if cycle $c$ is matched and $0$ otherwise
    \item $u_e$: $1$ if edge $e$ is matched and $0$ otherwise (not part of the original PICEF model).
\end{itemize}

Following this, we define additional parameters (aligning with those described earlier in the paper, as well as new formulation-specific parameters):
\begin{itemize}
    \item $E, P, N$: the set of edges, patient-donor pair vertices, and NDD vertices
    \item $H\subseteq E$: the set of half-compatible edges
    \item $h\in \mathbb Z_+$: the maximum number of immunosuppressants
    \item $w : E \to \mathbb R$: the edge weight for edge $e$
    \item $\mathcal K(e)\subseteq \{1, \dots, L\}$: the set of \emph{positions} that edge $e$ can take in a chain, where $K$ is the maximum chain length
    \item $C$: the set of feasible cycles (up to length $D$). With some abuse of notation, we denote membership in a cycle using ``$\in$'' for both edges and vertices: e.g., if edge $e$ is used in cycle $c$ then $e\in c$; if vertex $i$ participates in $c$, then $i \in c$.
\end{itemize}

Finally, we construct ILP~(\ref{eq:ilp}) below as follows.

{\small
\begin{equation}
    \begin{array}{rll} 
    \max & \sum\limits_{e\in E} u_e w(e)\\
    & \\
    & \sum\limits_{e\in \delta^-(i)}  \sum\limits_{k\in\mathcal K(e)} y_{ek} + \sum\limits_{\mathclap{{\tiny \begin{array}{c} c\in C: \\ i\in c\end{array}}}} z_c \leq 1 & \forall i\in P\\
    & \sum\limits_{{\tiny \begin{array}{c} e\in \delta^-(i)\land\\ k\in\mathcal K(e) \end{array}}}  y_{ek} \geq \sum\limits_{e\in \delta^+(i)} y_{e,k+1} &{\arraycolsep=0pt \begin{array}{l} \forall i\in P, \\ k\in \{1,...,L-1\} \end{array}}\\
    &\sum\limits_{e\in \delta^+(i)} y_{e1} \leq 1 & \forall i \in N \\
    & u_e = \sum\limits_{k \in \mathcal K(e)} y_{ek} + \sum\limits_{\tiny {\begin{array}{c} c\in C: \\ e\in c\end{array}}} & \forall e\in E\\
    & \sum\limits_{e\in H} u_e \leq h \\
    & y_{ek} \in \{0,1\} & \forall e\in E, k\in \mathcal K(e)\\
     &z_c \in \{0,1\} & \forall c\in C \\
     & u_e \in \{0, 1\} & \forall e\in E
     \end{array} \label{eq:ilp}
 \end{equation}
 }
The final two constraints are not part of the original PICEF model: the first new constraint defines variables $u_e$, which is $1$ if edge $e$ is matched; the second new constraint requires that at most $h$ half-compatible edges are matched.

This ILP model is powerful, because it can deal with bounded cycle length and a budget on the number of suppressants at the same time. Even though it does not provide a polynomial method to solve the problem (since such an algorithm cannot exist unless P=NP), ILP formulations have proved to work well in practical scenarios~\cite{CKVR13,BHA19,BKM+19}.

\section{Experimental Results}\label{sec:experiments}
In this section, we demonstrate the utility of immunosuppressants, with computational experiments on simulated kidney exchange graphs generated using data from the United Network for Organ Sharing (UNOS). For each UNOS graph, we begin with all vertices $V$ and fully-compatible edges $E_F$. Then we add new half-compatible edges by enumerating every blood-type-compatible pair of vertices that are \emph{not} already connected; we randomly create edges between fraction $\alpha\in [0, 1]$ of these pairs. Let these half-compatible edges be denoted by $E_H$; they can be matched only with an immunosuppressant. We denote the full set of edges as $E = E_F \cup E_H$. All edges have weight $1$.

For each graph we find the optimal matching by solving ILP~(\ref{eq:ilp}) with a \emph{budget} of $h \in \{0,\dots, 100\}$ suppressants.

\noindent\textbf{Results.}
%
%
For each immunosuppressant budget $h\in \{0, \ldots, 100\}$, we find the optimal matching by solving Problem~\ref{eq:ilp}; then, let $M_h$ denote the matching weight (objective value) of this optimal matching using at most $h$ suppressants.  Then, for each graph and each $h \in \{1, \ldots, 100\}$ we calculate
$\%Baseline \equiv 100 \times \frac{M_h - M_0}{M_0}.$
In other words, $\%OPT_h$ is the percentage-difference between the matching weight with budget $h$, and with budget $0$ (no half-compatible edges).
Figure~\ref{fig:lineplot} shows $\%Baseline$ for each set of random graphs, and for $\alpha \in \{0.05, 0.1, 0.2\}$.  Figure~\ref{fig:median-plot} shows the median percentage of each patient type matched.  Additional figures in the supplemental material give further information about the spread of results over all the simulated runs.

\begin{figure}[h!]
    \centering
    \includegraphics[width=0.75\linewidth]{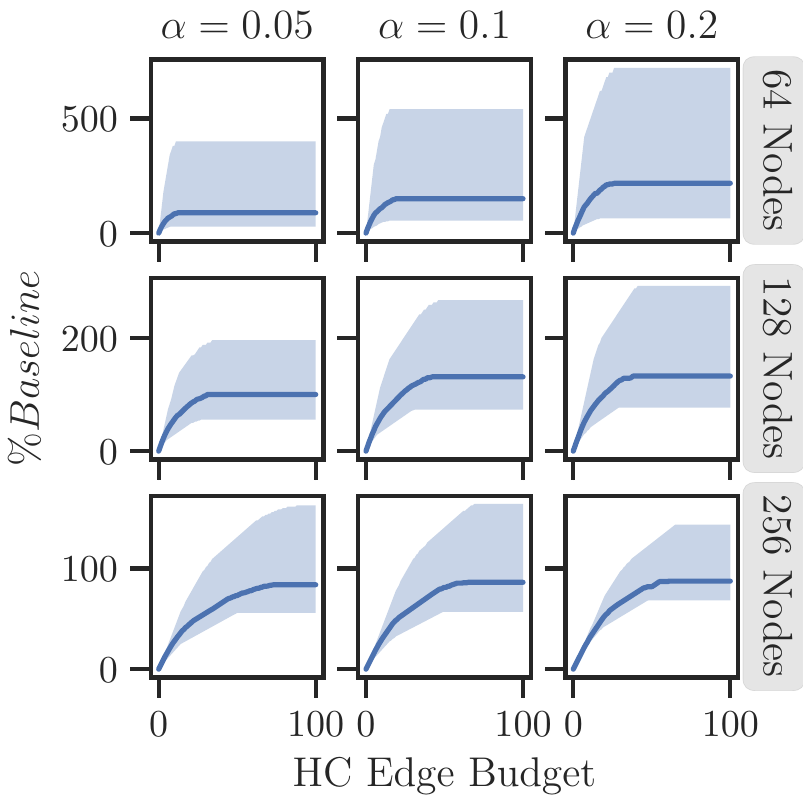}
    \caption{Median $\%Baseline$ for each set of graphs (top: 64-node graphs, middle: 128-node graphs, bottom: 256-node graphs), and each $\alpha\in \{0.05, 0.1, 0.2\}$.
    Shading is between the min and max values of $\%Baseline$.}
    \label{fig:lineplot}
\end{figure}

Figure~\ref{fig:lineplot} shows the immediate benefit of matching half-compatible edges.  Unsurprisingly, increasing the budget $h$ results in diminishing marginal returns; the greatest marginal benefit comes from the first 10 edges. For the small- and medium-sized (i.e., $64$ and $128$-node) graphs, that relatively small budget nearly \emph{doubles} the match size (weight); for the largest size (i.e., $256$-node) graphs, that relative gain is 50\% more---still a substantial gain.  There is also a trailing off effect such that, given enough immunosuppressant budget, no additional gain can be achieved.

\begin{figure}[h!]
    \centering
    \includegraphics[width=0.85\linewidth]{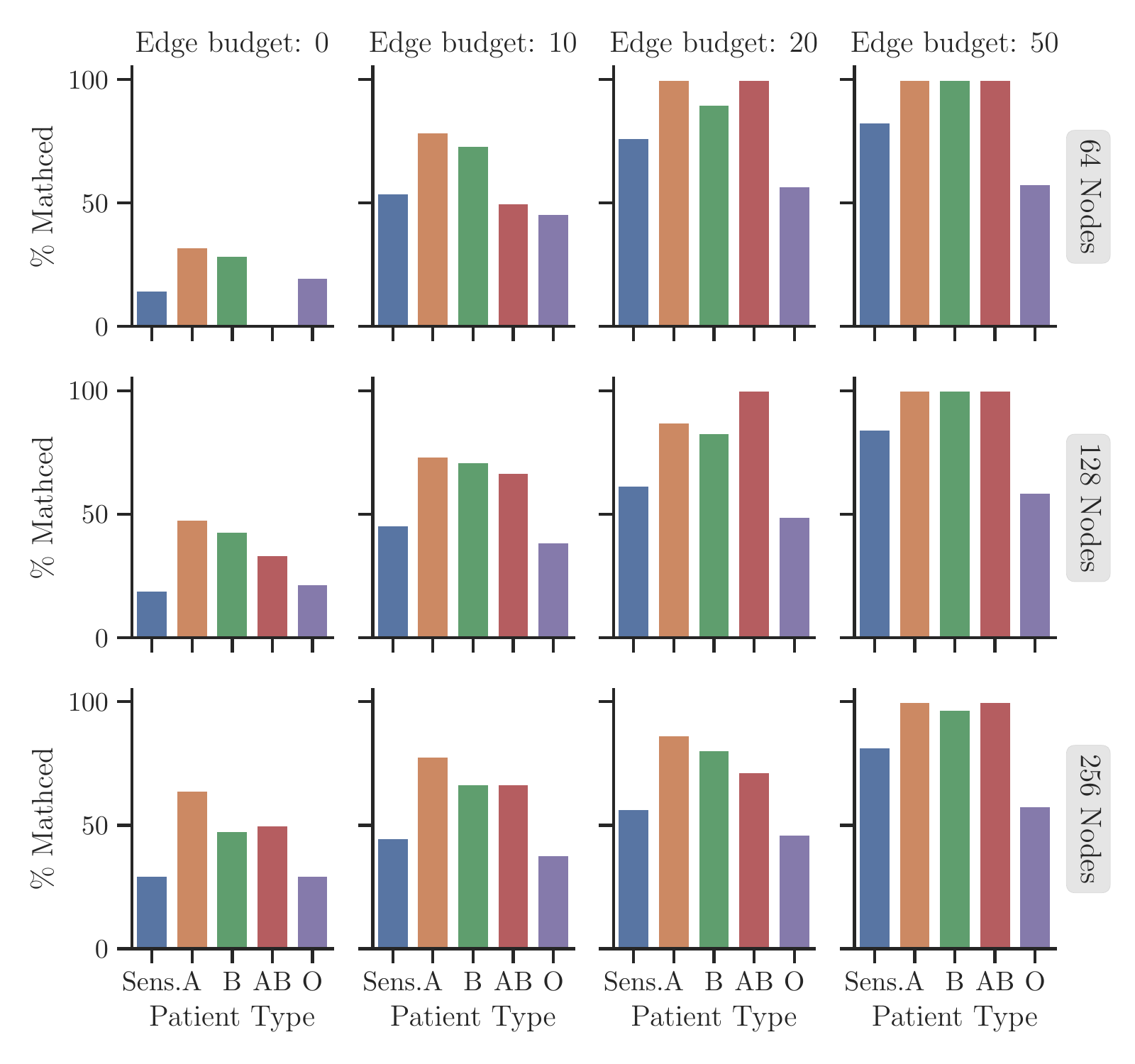}
    \caption{Median percentage of each patient type: highly-sensitized (Sens.), and blood type (A, B, AB, O), for each set of random graphs (top: 64-node graphs, middle: 128-node graphs, bottom: 256-node graphs), and $\alpha=0.2$.
    Each column shows a different edge budget (0, 10, 20, 50).}
    \label{fig:median-plot}
\end{figure}

Recall that we are only able to ``activate'' potential edges between blood-type-compatible vertices; thus, many pairs of vertices may \emph{never} be connected directly (e.g., O-type patients and AB-type donors), and graph structure may prevent vertices from ever being matchable at all.  Figure~\ref{fig:median-plot} shows this behavior: relatively more of the ``easier-to-match'' blood types (AB, A, and B) are matched than the O-type patients, i.e., those with the hardest-to-match blood type.  Still, Figure~\ref{fig:median-plot} shows that in aggregate patients of each blood type are helped---again, exhibiting diminishing marginal returns as immunosuppressant budget increases. 

Real-world kidney exchange pools range in size from a few dozen patient-donor pairs and altruistic donors---either at individual transplant centers or in burgeoning but still-nascent multi-center programs---to a few hundred in larger exchanges in the US, UK, and (soon) multinational exchanges.  Our experimental results support that application of even a small number of suppressants results in large gains on realistic kidney exchange graphs of varying, realistic size.



\clearpage
\section*{Ethical Implications of our Work}
Kidney exchanges save lives and are broadly viewed as beneficial to humanity; however, as in many resource-constrained settings, decision-makers must make morally-laden decisions when designing the objective functions, constraints, and other modeling concerns that increasingly run modern exchange programs.  The economics, AI, operations research, bioethics, medical, and legal communities have long discussed the moral implications of different approaches to the allocation of organs~\citep[see, e.g.,][]{Cohen89:Increasing}, including kidney exchanges~\citep[see, e.g.,][]{Ross97:Ethics,Minerva19:Ethics,Torres19:Bi-organ}.  Broadly speaking, our proposed work falls into the category of creating a more general, and thus potentially more powerful, model for the exchange of organs, and thus may come with many of the same positive and negative potential ethical impacts.  Positives are clear: those who could not previously receive a kidney may now be afforded that opportunity, and those who would have been matched to a relative worse kidney donor are now afforded the opportunity to match to a relatively better one.  Specific to our model, though, is the potential ethical implication of applying a suppressant to one patient so that another patient---matched elsewhere in a cycle or chain---might receive a kidney.  There is a cost---both monetary and in terms of quality of health---to immunosuppression; thus, an open and morally-laden question lies in determining the tradeoffs between, and level of agency given to, participants in exchanges that run immunosuppression schemes.  As in many such scenarios, there is no ``globally correct'' answer, but rather only an answer that can be arrived at after careful consideration by stakeholders: patients, donors, doctors, ethicists, lawyers, and possibly others.  We do not prescribe a specific solution here, but rather note that our model is general and could, with input from domain experts, be augmented to address some of these concerns.

\noindent\textbf{Acknowledgements.}   
Cseh was supported by the Hungarian Academy of Sciences under its Momentum Programme (LP2016-3/2020), OTKA grant K128611, and COST Action CA16228 European Network for Game Theory. Dickerson and McElfresh were supported in part by NSF CAREER Award IIS-1846237, NSF Award CCF-1852352, NSF D-ISN Award \#2039862, NIST MSE Award \#20126334, NIH R01 Award NLM-013039-01, DARPA GARD Award \#HR00112020007, DoD WHS Award \#HQ003420F0035, DARPA Disruptioneering Award (SI3-CMD) \#S4761, and a Google Faculty Research Award. 
\bibliography{abbshort,mybib,group}


\end{document}